\documentclass[conference]{IEEEtran}
\makeatletter

\def\ps@headings{%
\def\@oddhead{\mbox{}\scriptsize\rightmark \hfil \thepage}%
\def\@evenhead{\scriptsize\thepage \hfil \leftmark\mbox{}}%
\def\@oddfoot{}%
\def\@evenfoot{}}
\makeatother
\usepackage{cite}
\usepackage{amsmath,amssymb,amsfonts}
\usepackage{graphicx}
\usepackage{textcomp}
\usepackage{cuted}
\usepackage{balance}
\usepackage{algorithm}
\usepackage{algpseudocode}
\usepackage{url}
\usepackage{amsthm}

\newtheorem{theorem}{Theorem}
\newtheorem{lemma}{Lemma}

\usepackage{mathtools}

\usepackage[table]{xcolor}    % loads also »colortbl«
\usepackage[para]{threeparttable}
\usepackage{times}  % For BA*

\usepackage{color}
\definecolor{mycolor}{rgb}{0.5,0,1}

\usepackage{caption}
\usepackage{subcaption}

\begin{document}

% paper title
\title{On the Feasibility of Sybil Attacks in \\Shard-Based Permissionless Blockchains}

\author{Tayebeh Rajabi$^{\dag}$, Mohammad Hossein Manshaei$^{\ddag \dag }$, Mohammad Dakhilalian$^{\dag}$, \\ Murtuza Jadliwala$^\diamond$, and Mohammad Ashiqur Rahman$^\ddag$\\ %Murtuza Jadliwala$^\diamond$, and \\
$^\dag$Department of Electrical and Computer Engineering, Isfahan University of Technology, Iran\\
$^\diamond$Department of Computer Science, University of Texas at San Antonio, USA\\
$^\ddag$Department of Electrical and Computer Engineering, Florida International University, USA\\
Emails: t.rajabi@ec.iut.ac.ir, \{manshaei,mdalian\}@iut.ac.ir, murtuza.jadliwala@utsa.edu, marahman@fiu.edu.}

% make the title area
\maketitle

\begin{abstract}
%A new way was beginning to extend and improve distributed networks, by introducing Bitcoin and Blockchain.
%Security and scalability are the most important issues for change centralized networks to distributed networks with Blockchain. 
Bitcoin's single leader consensus protocol (Nakamoto consensus) suffers from significant transaction throughput and network scalability issues due to the computational requirements of it Proof-of-Work (PoW) based leader selection strategy. To overcome this, committee-based approaches (e.g., \emph{Elastico}) that partition the outstanding transaction set into shards and (randomly) select multiple committees to process these transactions in parallel have been proposed and have become very popular. However, by design these committee or shard-based blockchain solutions are easily vulnerable to the \emph{Sybil attacks}, where an adversary can easily compromise/manipulate the consensus protocol if it has enough computational power to generate multiple Sybil committee members (by generating multiple valid node identifiers). 
%In this attack, the adversary tries to impersonate and makes many Sybil IDs allowing his/her to cast many votes to increase his/her influence.
%in the network by casting rigged votes.
%In this way, the adversary can interfere consensus in some shards, even can change their outputs or add invalid transactions to the blockchain. 
Despite the straightforward nature of these attacks, they have not been systematically analyzed. In this paper, we fill this research gap by modelling and analyzing Sybil attacks in a representative and popular shard-based protocol called \emph{Elastico}. We show that the PoW technique used for identifier or ID generation in the initial phase of the protocol is vulnerable to Sybil attacks, and a node with high hash-power can generate enough Sybil IDs to successfully compromise Elastico. 
%In summary, the attacker can control a shard by generating enough Sybil IDs. 
We analytically derive conditions for two different categories of Sybil attacks and perform numerical simulations to validate our theoretical results under different network and protocol parameters. 
%Our analytical results conclude that a node with 25\% hash-power of the entire network \remark{"of the entire network"? The minimum requirement?} can generate 25\% of the total IDs of the network, leading to a successful launch of this attack. %So the PoW is not a suitable way for ID generation in the Shard-based protocols.
\end{abstract}

%\remark{There are inconsistent capitalization.}
%%%%%%%%%%%%%%%%%%%%%%%
\section{Introduction}

%S. Nakamoto introduced the concept of blockchain by proposing Bitcoin in 2008~\cite{nakamoto2008bitcoin}. It uses a distributed ledger managed by a Peer to Peer (P2P) network to securely store the data and transactions. It is generally a permissionless infrastructure, and there exists neither a centralized controller nor a Trusted Third Party (TTP). %\footnote{Trusted Third Party} Each node of this system can see this ledger. If the node has the permission as well as enough hash-power, it can add transactions on the blockchain. The blockchain solves the agreement in the distributed network by a consensus protocol and supports new applications that require distributed computing, e.g., Cryptocurrency, Smart Contracts, and Internet of Things (IoT)~\cite{kiayias2016blockchain}. The blockchain technology has six key properties: decentralized, transparency, anonymity, open-source, immutability, and autonomy~\cite{lin2017survey}.
A blockchain is an append-only and immutable distributed database or ledger that records a time-sequenced history of transactions. One key aspect of the blockchain protocol is the consensus algorithm which enables agreement among a network of \emph{nodes} or \emph{processors} on the current state of the ledger, assuming that some of them could be malicious or faulty. Blockchain protocols are classified as permissioned or permissionless depending on whether a trusted infrastructure to establish verifiable identities for nodes is present or not. One of the first, and still popular, permissionless blockchain protocol is Bitcoin~\cite{nakamoto2008bitcoin}. Consensus in Bitcoin is achieved by selecting a leader node in an unbiased fashion once every 10 minutes (an epoch) on an average, who then gets the right to append a new block onto the blockchain. The network then implicitly accept this block by building on top of it in the following epoch or reject it by building on top of some other block. Bitcoin uses a Proof-of-Work (PoW) mechanism to select the leader in each epoch. In Bitcoin's PoW, each node independently attempts to solve a hash puzzle and the one that succeeds is selected as a leader who gets the right to propose the next block. As PoW involves significant computation, Bitcoin's protocol includes a reward mechanism for the winning node in order to incentivize everyone to behave honestly. Ever since the introduction of Bitcoin in 2008, the power of permissionless blockchain technology has been harnessed to create systems that can host and execute distributed contracts (also referred to as ``smart contracts") which have found many interesting applications, including, cryptocurrencies, secure data sharing and digital copyright management to name a few~\cite{kiayias2016blockchain}.

%Maintaining a blockchain requires huge computational power. It is also not scalable. Bitcoin can process up to seven transactions per second~\cite{WikiScaleBitcoin}\remark{It can be up to 10.7 \url{https://geoprise.com/blog/blockchains-capacity-limited-7-transactions-second}}. This transaction throughput is significantly low compared to some centralized networks like Visa that processes 56000 transactions per second \cite{manshaei2018game}. In 2016, \emph{Sharding} was introduced as a solution to this low scalability issue. This proposed shard-based blockchain is named as \emph{Elastico}~\cite{luu2016secure}. \emph{Omniledger}~\cite{kokoris2018omniledger} and \emph{Zilliqa}~\cite{secure2018zilliqa} are another two shard-based blockchains that were introduced in 2018. 
A significant shortcoming of Bitcoin's (and of other similar permissionless system's) leader and PoW competition based consensus protocol is its low transaction throughput and poor network scalability. For instance, Bitcoin's and Ethereum's transaction rates are only 7 and 20 transactions per second, respectively, which is clearly not sufficient for practical applications \cite{bitcoinscalability}. Although there have been several efforts towards improving the Bitcoin protocol itself, for instance, BIP \cite{bip102} and Bitcoin-NG \cite{eyal2016bitcoin}, other works have focused towards developing alternate high throughput and scalable permissionless blockchain protocols. One key outcome in this line of research is committee or shard-based protocols ~\cite{luu2016secure,kokoris2018omniledger,secure2018zilliqa} that operate by periodically partitioning the network of nodes into smaller non-overlapping committees, each of which processes a disjoint set of transactions (also called, a shard) in parallel with other committees. As each committee is reasonably small, the transaction throughput and scalability can be significantly improved by running a classical Byzantine consensus protocol such as PBFT \cite{castro1999practical} within each committee (and in parallel across all committees) rather than the traditional PoW based competition as used in Bitcoin. The idea of parallelizing the tasks of transaction processing and consensus by partitioning the processor network into committees is very promising and is already seeing deployment \cite{ethereumupgrade}.

%%%%%%%%%%%%%%%
\begin{figure*}[t]
\centering
\includegraphics[scale= 0.29]{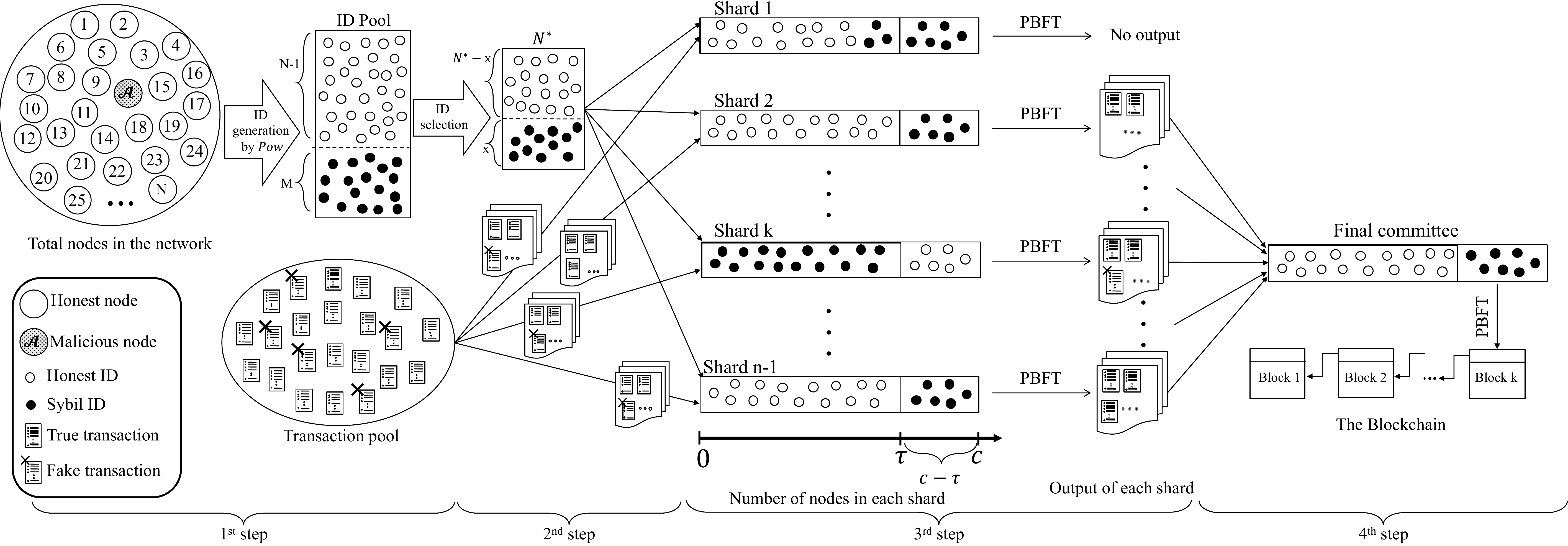}
\caption{{\small A shard-based blockchain system model. \emph{Step 1:} All nodes generate their IDs using a PoW defined mechanism. \emph{Step 2:}  IDs and transactions will be distributed among shards, given the generated IDs. There exist $2^{s}$ shards. Each shard accepts up to $c$ IDs.  \emph{Step 3:} Each shard runs a PBFT mechanism to validate the assigned transaction sets. At least $\tau$ IDs have to agree on the output of each shard. \emph{Step 4:} Final committee  merges the agreed values of shards and create a final block, that will be added to blockchain. In this example, the attacker breaks the consensus protocol in Shard $1$ (i.e., BCP attack) and inserts fake transactions in the output of Shard $k$ (i.e., GFT attack). }}
\label{fig:SysModel}
\end{figure*}
%%%%%%%%%%%%%

%A shard-based protocol does not use predefined or static IDs for the nodes. Instead, each node has to generate an ID for itself in each processing epoch. However, the ID generation phase requires to be Sybil attack-resistant. This is because, if a node can create many IDs in an epoch, it can vote to the transaction validation phase multiple times than other nodes. It will change the blockchain network to an unfair network. The malicious node with a large number of IDs can ignore some true transactions or add some fake transactions to the blockchain. Therefore, Sybil attacks can damage these shard-based permissionless blockchain networks severely. 
%
%\remark{What do you want to say/claim in the following sentence? Reference?} These protocols assume that all nodes have the same hash-power, and they propose the Proof of Work (PoW) mechanism as a Sybil-resistant mechanism in the ID generation phase.
Committee participation in popular shard-based protocols such as Elastico~\cite{luu2016secure} happen by nodes generating verifiable IDs using some pre-defined PoW puzzle at the beginning of each epoch - only nodes that are able to generate valid IDs are able to participate in the consensus process in that epoch. However, it should be easy to see that the ID generation process easily lands itself to a Sybil attack. An adversary (valid node in the network) can leverage its computational or hash-power to generate a large number of Sybil IDs and increase its participation in shards. The adversary, by means of the generated Sybil IDs, can compromise the intra-committee consensus protocol to either prevent new transaction blocks from being added or to add fake transactions to the blockchain. Existing shard-based protocols assume that all nodes have the same hash-power, and thus the PoW based ID generation mechanism with an appropriate difficulty is not prone to such Sybil attacks. 
%MJ-Please make sure the above sentence/claim is true. If needed appropriate modify it and add a citation here
This assumption about the uniformity of hash-power across nodes generally does not hold, making such Sybil attacks feasible. Despite this, there has been very little effort within the research community to analytically study and combat the threat of Sybil attacks in shard-based blockchains. In this paper, we attempt to address this research gap.

%Specifically, we investigate two different types of Sybil attacks in shard-based blockchain systems and theoretically analyze their feasibility in terms of their success probability given certain system/network parameters. We mathematically model a Sybil attack and present a set of equations that calculate the probability of this attack. We consider a practical context in which different nodes of the network can have different hash-powers. In an attack scenario, we assume that some nodes have significantly high hash-power. These nodes can generate many IDs in an epoch, and can get the control of the shards using these Sybil IDs. 
%
%In this research, we consider the Elastico as the case study. 
%and model this blockchain network. 
%We assume that the malicious node who has enough hash-power, generates many Sybil IDs. 
%We model this blockchain network and along with a Sybil attack on this network.  
%We investigate the possibility of this attack analytically. 
%Our analysis shows that the probability of this attack depends on various network/shard parameters. 
%\remark{Should we write about our simulation?}
%We present some analytical results that can help us to design a network that can resist Sybil attacks while maintaining the scalability provided by the shards.
%In this situation the blockchain becomes scalable. Also, it becomes fair and secure.

Specifically, we investigate two different types of Sybil attacks in a representative shard-based blockchain system such as Elastico: (i) Break Consensus Protocol (or BCP) attack where the goal of the adversary is to thwart the consensus process %and prevent addition of new transaction blocks, 
and (ii) Generate Fake Transaction (or GFT) attack where the goal of the adversary is to introduce fake or invalid transactions into the blockchain. By assuming a reasonable network and adversary model, we first derive bounds on the success probability of these attacks and then theoretically analyze the necessary conditions to achieve success in these attacks. We further validate our analytical results by means of numerical simulations for a variety of system and network parameters. These analytical results and simulations shed further light on the computational or hash-power requirement for an adversary to compromise the consensus protocol of shard-based protocols and the choice of system and network parameters that can significantly reduce the probability of such attacks. The remainder of this paper is structured as follows.

%MJ-Some more specific details of our simulations (specific numbers) would help here.
 In Section~\ref{sec:SYSModel}, we introduce our system model for a Sybil attack to shard-based blockchain. In section~\ref{sec:analytical}, we present our analytical results for the probability of a successful Sybil attack. Section~\ref{sec:numerical} presents the simulation results, following by related work and conclusions in Section~\ref{sec:ralated work} and Section~\ref{sec:conclusion}.

%%%%%%%%%%%%%%%%%%%%%%%
\section{System Model}
\label{sec:SYSModel}

In this section, we first briefly outline the operational details of a shard-based permissionless blockchain system such as Elastico and then discuss the details of Sybil attacks that may be possible in such a network.

\subsection{Shard-Based Blockchain Models}
The operation of Elastico, and most other shard-based permissionless blockchain networks, can be fully described by means of four sequential steps that needs to be executed in each time period or epoch, as outlined in Figure~\ref{fig:SysModel}.
%Figure~\ref{fig:SysModel} shows a simple representation of Elastico, the first shard-based blockchain network in four steps. 
We assume that the blockchain network comprises of $N$ nodes (or processors) with different computational capabilities or hash-power. This is in contrast with \cite{kokoris2018omniledger} and \cite{luu2016secure}, where it is assumed that all nodes have the same power. The hash-power, which specifies the number of hash computations that can be performed by a node per second, is denoted by $h$. To elucidate the presentation, Table~\ref{vartable} summarizes the symbols (and their meanings) used throughout the paper.
%, and each node has to generate one ID to enter the system. 
%We remark that some of nodes are malicious but we assumed that there is just one attacker in this model. 
%\remark{rewrite considering the definition of IDs and how we calculate them in Elastico.}

%
%%%%%%%%%%%%%%
\begin{table}[t]
\caption{List of Notations}
\vspace{-6pt}
\begin{center}
\rowcolors{2}{gray!20}{white}
\begin{tabular}{l |p{6.5cm}}
\hline
\textbf{Symbol} &\textbf{Definition} \\
\hline
\hline
$N$ & Number of nodes in the network \\
$N^{*}$ & Total required IDs at each epoch ($N^{*} <N$)  \\
$M$ &  Number of IDs generated by the adversary \\
$t_i$ & The time of epoch $i$\\
$L$ & The length of output of secure hash function (bit)\\
$L^{t_i}$ & The length of target value (bit) at epoch $t_i$\\
$c$ & Capacity of each shard \\
$h$ &  Hash-power of each processor\\
$s$ & Represent the number of shards ($2^{s}$) \\
$T_I$ & Initialization time needed for ID generations  \\
$\tau$  &  Consensus threshold \\ 
$x$ & Number of chosen adversary's IDs (Random Variable)  \\
$Dif(t_i)$ & Difficulty of solving PoW puzzle at epoch $t_i$ \\
$p(t_i)$ & The probability of finding a correct ID at epoch $t_i$\\
$P\{x=n\}$ & The probability of selecting $n$ adversary's ID\\
$P_{c-\tau +1}$ & The probability of having at least $c-\tau +1$ adversary's IDs in one shard, if $n \in (c-\tau ,c]$ \\
$P'_{c-\tau +1}$ & The probability of having at least $c-\tau +1$ adversary's IDs in one shard, if $n \in (c,2^{s}(c-\tau)]$ \\
$P''_{c-\tau +1}$ & The probability of having at least $c-\tau +1$ adversary's IDs in one shard, if $n > 2^{s}(c-\tau)$ \\
$P_{\tau}$ &The probability of having at least $\tau$ adversary's IDs in one shard, if $n \in [\tau ,c]$    \\
$P'_{\tau}$ & The probability of having at least $\tau$  adversary's IDs in one shard, if $n \in (c,2^{s}(\tau -1)]$  \\
$P''_\tau$ & The probability of having at least $\tau$ adversary's IDs in one shard, if $n > 2^{s}(\tau -1)$ \\
$P_B$ & The probability of successful BCP attack\\
$P_G$ & The probability of successful GFT attack\\
\hline
\end{tabular}
\end{center}
\label{vartable}
\vspace{-6pt}
\end{table}
%%%%%%%%%%%%

Similar to any permissionless system, nodes in Elastico do not have any predefined identity assigned by a trusted third-party. In the first step of Elastico, as shown in Figure\ref{fig:SysModel}, each node attempts to generate a verifiable and pseudorandom ID which will enable it to participate in the rest of the steps in that time period or epoch. The nodes use the solution of a Proof-of-Work (PoW) puzzle with a network-determined difficulty to decide if they have arrived at a valid ID, as described next. Let $Hash()$ be the hash function (Elastico employs $SHA-256$) employed by a node in the blockchain network and let $IP$ and $PK$ denote a node's network address and public key, respectively. %The length of output of the hash function is $L(bit)$. 
A publicly-known pseudo-random string \emph{epochRandomness}, generated at the end of the previous epoch of the protocol (to avoid puzzle pre-computation), is used as a seed for the PoW puzzle.
% to ensure that the PoW is not precomputed. 
Each node in the first step of the protocol attempts to solve the PoW puzzle by finding a nonce such that $Hash(epochRandomness||IP||PK||nonce)$ is smaller than some network-determined \emph{target} value. The target value which determines the difficulty of the PoW puzzle is adapted by the network during each epoch based on the network-wide hash-power. Let's denote the target by $L^{t_i}$ bits, where $L$ is the size of the message digest (in terms of bits) and $t_i$ is corresponding the time epoch. In other words, a valid ID value during epoch $t_i$ must be smaller than $2^{L^{t_i}}$ and a node successful in generating such a valid ID assumes it as its own ID during the later steps of the the protocol. It should be noted that all nodes must generate their ID during a given \emph{initialization time} $T_I$, defined by the protocol. If they cannot solve the puzzle within this time, they will not possess a valid ID to join the rest of the network and participate in the protocol. It should be clear that nodes with a higher hash-power have a higher probability of solving the ID generation PoW puzzle and thus participating in the protocol, compared to nodes with lower hash-power.
%)  could solve this problem faster than other nodes. 

%As shown in the Figure~\ref{fig:SysModel}, 
After ID generation, in step 2 the generated node IDs and distributed transactions (which may contain both valid and fake/invalid transactions) are randomly distributed (or partitioned) into different \emph{shards} or committees for validation. In Elastico, there exist a total of $2^{s}$ shards, where $s$ is a network-defined parameter. Each node will be placed in a shard according to the last $s$ bits of its ID. If the \emph{capacity} of each shard is denoted by $c$, then it is clear that the minimum number of valid IDs required to execute a single epoch (all steps) of a shard-based blockchain protocol such as Elastico is $N^* = 2^{s}\times c$. The processors discover IDs of other processors in their shard by communicating with each other. %A set of transactions (containing true or fake transactions) will be assigned to each shard for possible validation. 

In the third step, processors of each shard simultaneously validate the transaction set assigned to that shard and agree on a consensus transaction set (within that shard) using the PBFT \cite{castro1999practical} and \cite{buchman2016tendermint} algorithm.
%MJ-use the full form of PBFT above and include citation
Let $\tau$ denote the consensus threshold for each shard, i.e., if at least $\tau$ processors in a shard agree on a transaction set, then consensus within the shard is successful and the consensus transaction set is added to the shard's output. In other words, the consensus protocol within a shard successfully outputs a valid consensus transaction set even if $c - \tau$ nodes within the shard do not cooperate and/or behave maliciously.
%. In other words, if no more than $c-\tau $ members of every shard are malicious, the consensus protocol work and shard will have a valid output. 
In Figure~\ref{fig:SysModel}, we see this case in Shards $2$ and $n-1$. Post the intra-shard consensus, the leader node within each shard sends the signed value of the consensus transaction set, along with the signatures of the contributing shard nodes, to a final consensus committee (step 4). The final committee of processors, chosen based on their ID, merges the consensus transaction sets from each shard to create a final block which is eventually appended to the blockchain. %This shard can verify that a certain value is the selected one by checking that it has sufficient signatures. 
Each final committee member first validates that the values (consensus transaction sets) received from each shard is signed by at least $c/2 + 1$ members of the shard, and then computes the ordered set union of all inputs. Finally, nodes of final committee execute PBFT to determine the consensus final block, which is signed and gets appended as the next block to the blockchain. 

\subsection{Threat Model }
\label{subsec:threatmodel}

We assume that the adversary or attacker ($\mathcal{A}$) in this setup has enough hash-power to generate more than one IDs (during step 1, as outlined above) and attempts to launch a \emph{Sybil attack} to disrupt the operation of the shard-based blockchain protocol. Although we formally model adversarial capabilities assuming a shard-based blockchain protocol such as Elastico, the model (and the results) can easily be generalized to other shard-based protocols. Before outlining the specific Sybil attacks that could be carried out by the adversary, let us first characterize the difficulty of an adversary in generating Sybil nodes. 
%Elastico deploys a PoW technique for the ID generation phase. 
As outlined earlier, each node or processor uses the solution of a PoW puzzle as an ID such that in epoch $t_i$ the selected ID must be smaller than some network-agreed target value $2^{L^{t_i}}$. Let $MaxTarget$ denote the maximum possible value of the target. As $MaxTarget$ is observed in the first epoch, $MaxTarget = 2^{L^{t_1}}$ (e.g., in Bitcoin $MaxTarget = 2^{224}$ \cite{o2014bitcoin}). Given $MaxTarget$, we define the \emph{difficulty} of solving a PoW puzzle in epoch $t_i$ as:
\begin{equation}\label{eq:Diff}
%Dif(t_i)=\frac{MaxTarget}{Current\ Target}=\frac{2^{L^{t_1}}}{2^{L^{t_i}}}.
Dif(t_i)=\frac{MaxTarget}{2^{L^{t_i}}}.
\end{equation}
It is easy to see that, as the target value $2^{L^{t_i}}$ during a particular time epoch $t_i$ reduces, the PoW puzzle for ID computation becomes harder to solve for the nodes or processors, which is indicated by a higher difficulty value $Dif(t_i)$. 
%The $MaxTarget$ and the difficulty of the PoW is usually determined by the type of blockchain. 
Now, the probability of finding a valid ID during epoch $t_i$ is given by:
\begin{equation}\label{eq:pti}
p(t_i)=\frac{2^{L^{t_i}}}{2^{L}}. %\frac{Current Target}{2^{L}}=
\end{equation}
%In fact,  $2^{L^{t_i}}$ designates the number of solutions that is smaller than $Current Target$, and 
where $2^{L}$ denotes the message digest space of the hash function $Hash()$ using the PoW puzzle. Given Equation~(\ref{eq:Diff}), $p(t_i)$ can be rewritten as:
\begin{equation}\label{pID1}
p(t_i)=\frac{\frac{2^{L^{t_1}}}{Dif(t_i)}}{2^{L}}=\frac{2^{L^{t_1}}}{Dif(t_i)\times 2^L}.
\end{equation} 
Equation~(\ref{pID1}) above represents the formal relationship between the difficulty of finding a valid PoW puzzle solution (i.e., ID in this case) and the success probability of solving the puzzle during epoch $t_i$.
% given the message digest space ($L$). 
Now, let us assume that the adversary $\mathcal{A}$'s hash-power (or hash computation capability) is $h^{\mathcal{A}}$. In other words, during the initialization period $T_I$ (step 1) of an epoch $t_i$, $\mathcal{A}$ can generate a maximum of $h^{\mathcal{A}} \times {T_I}$ potential solutions (or message digests), of which only those that smaller than $2^{L^{t_i}}$ (target value during $t_i$) can be used as valid IDs.
%MJ-start here
Let's further assume that $\mathcal{A}$ can find $M$ valid solutions or IDs (i.e., those that satisfy the puzzle or are within the target value) during $T_I$. Thus, $M$ can be calculated by:
\begin{equation}\label{MID}
M =  p(t_i) \times {h^{\mathcal{A}} \times T_I}.
\end{equation} 
Substituting for $p(t_i)$ from Equation~(\ref{pID1}) we get:
\begin{equation}\label{MID2}
M = \frac{2^{L^{t_1}}}{Dif(t_i)\times 2^L} \times {h^{\mathcal{A}} \times T_I}.
\end{equation}
Now, a PoW puzzle based identity or ID generation mechanism is said to be \emph{strictly Sybil-resistant} if and only if the value of $M$ in the mechanism can be restricted to less than 2. In other words, for the above PoW puzzle based identity or ID generation mechanism in Elastico to be strictly Sybil-resistant, the solution difficulty $Dif(.)$ and initialization time $T_I$ should be such that for a given adversary hash-power $h^{\mathcal{A}}$, $L$, and $L^{t_1}$, $M$ is always smaller than 2, i.e.,  $\mathcal{A}$ should be able to generate at most one ID during $T_I$. In this paper, we assume that the shard-based blockchain protocol's PoW-based identity generation mechanism is not strictly Sybil-resistant and that the adversary's hash-power is large enough to generate two or more than IDs (i.e., $M \geq 2$) during the initialization phase (step 1) of any epoch. The adversary then employs these numerous valid IDs for placing itself in multiple shards in order to carry out different types Sybil attacks (as described below), the goal of which is to subvert the correct operation of the blockchain protocol. 
In contrast to the adversary, we assume that each honest node in the network generates only a single valid ID during $T_I$ in each epoch. Now, lets describe in further details the two different types of Sybil attacks that can be carried out by the adversary.
% according to the number of Sybil IDs that have been placed in a shard;

\noindent {\bf 1. Break Consensus Protocol (BCP) Attack:} In order to accomplish the BCP attack, the goal of which is to disrupt the shard-based consensus process, the adversary will need to generate more than $c-\tau$ valid IDs in a (target) shard. This threshold of valid IDs will enable the adversary to break the intra-shard consensus protocol in that shard, thereby preventing insertion of some transactions (specifically, from the target shard) into the blockchain. An instance of the BCP attack is depicted in Figure~\ref{fig:SysModel}, where the attacker successfully inserts more than $c-\tau$ Sybil IDs in Shard 1. 

\noindent {\bf 2.  Generate Fake Transaction (GFT) Attack:} In order to accomplish the GFT attack, the goal of which is to include fake (including, double spending or invalid) transactions in the blockchain blocks by taking over the consensus process, the adversary will need to add at least $\tau$ valid IDs in a (target) shard. By doing so, the adversary aims to control and manipulate the consensus process (i.e., the PBFT algorithm) using his Sybil IDs so that intra-shard consensus could be arrived on a desired set of fake transactions, which eventually get inserted into the blockchain block after final consensus. Figure~\ref{fig:SysModel} illustrates the GFT attack on shard $k$.

Our overarching goal is to determine bounds on the success probabilities, given different network and protocol parameters, of carrying out the BCP and GFT Sybil attacks described above.
% for BCP and GFT scenarios.
%
In the following section, we present theoretical analysis outlining the computation of these probability bounds.
% compute the probability of successful attacks for the aforementioned cases, given different parameter values in shard-based blockchains. 

\section{Analytical results}
\label{sec:analytical}

%In this section, we will compute the probability of successful BCP and GFT attacks in shard-based blockchains.
The remaining operations (step 2 onward) of the shard-based blockchain protocol are initiated in each epoch only after receiving $N^*$ IDs from the ID pool generated during the initialization or identifier generation phase (step 1), as discussed earlier in Section \ref{sec:SYSModel}. Now, given the $M$ Sybil IDs (generated by the adversary) in the ID pool comprising of a total of $M+N-1$ IDs, our first task is analyze this ID selection. Let $x$ denote a random variable representing the number of Sybil IDs (generated by the adversary) chosen from the $M+N-1$ IDs generated during the initialization step (step 1). In other words, $x$ IDs belonging to the adversary while $N^*-x$ IDs belonging to the honest nodes are distributed among the various shards after the initialization period $T_I$ (step 1). The success probability of the various attacks described above is thus a function of the number of Sybil IDs that the adversary is able to generate and get distributed among the different shards in each epoch. From the discussion in Section~\ref{subsec:threatmodel}, it should be clear that if the number of Sybil IDs generated by an adversary is smaller than or equal to $c-\tau$, then the adversary cannot successfully execute the BCP or GFT attacks. 
Let $n$ denote the actual number of Sybil IDs chosen to be distributed among the various shards, i.e., $n \in \{ 1,2, \cdots, M \}$. Thus, we first need to calculate the probability $P\{x=n\}$ of selecting/choosing $n$ Sybil IDs or nodes from the entire pool of $M+N-1$ IDs in an epoch.
%The total number of IDs that have been generated at an epoch is equal to $M+N-1$ (There exist $N-1$ honest IDs and $M$ Sybil IDs). Let us assume that $Prob\{x=n\}$ designates the probability of choosing $n$ Sybil IDs from the ID pool. 
This is given by the following Lemma. 
\begin{lemma}
If $\mathcal{A}$ can generate $M$ Sybil IDs during $T_I$, then the probability of selecting $n$ Sybil IDs from the ID pool is 
\begin{equation}\nonumber
P\{x=n\}=\frac{{\binom Mn} {N-1 \choose N^*- n}}{{M+N-1 \choose N^*}}.
\end{equation}
\label{pr}
\end{lemma}
%\vspace{-0.5cm}
\begin{proof}
The proof of this lemma follows trivially from the fact that $x$ follows a hypergeometric distribution \cite{ross2014first}.
\end{proof}
The following subsections will be devoted to the computation of the adversary's success probability in executing BCP and GFT Sybil attacks. Recall that Table~\ref{vartable} summarizes the definition of main probability bounds. 
% ,  we will calculate the probability of successful attacks based on the number of attacker's IDs ($M$) and other blockchain parameters. 
%We summarize the definition of main probability bounds in Table~\ref{analytical variable}.
% summarizes all calculated probabilities in our presented model.
%
%%%%%%%%%%%%%%%%
%\begin{table}[t]
%\caption{List of Analytical Variables.}
%\begin{center}
%\vspace{-6pt}
%\rowcolors{2}{gray!20}{white}
%%\setlength{\extrarowheight}{1.5pt}
%\begin{tabular}{l |p{6.5cm}}
%\hline
%\textbf{Symbol} & \textbf{Definition} \\
%\hline
%\hline
%%$n$ & The number of selected attacker's IDs\\
%$P\{x=n\}$ & The probability of selecting $n$ adversary's ID\\
%$P_{c-\tau +1}$ & The probability of having at least $c-\tau +1$ adversary's IDs in one shard, if $n \in (c-\tau ,c]$ \\
%$P'_{c-\tau +1}$ & The probability of having at least $c-\tau +1$ adversary's IDs in one shard, if $n \in (c,2^{s}(c-\tau)]$ \\
%$P''_{c-\tau +1}$ & The probability of having at least $c-\tau +1$ adversary's IDs in one shard, if $n > 2^{s}(c-\tau)$ \\
%$P_{\tau}$ &The probability of having at least $\tau$ adversary's IDs in one shard, if $n \in [\tau ,c]$    \\
%$P'_{\tau}$ & The probability of having at least $\tau$  adversary's IDs in one shard, if $n \in (c,2^{s}(\tau -1)]$  \\
%$P''_\tau$ & The probability of having at least $\tau$ adversary's IDs in one shard, if $n > 2^{s}(\tau -1)$ \\
%$P_B$ & The probability of successful BCP attack\\
%$P_G$ & The probability of successful GFT attack\\
%\hline
%\end{tabular}
%\end{center}
%\vspace{-6pt}
%\label{analytical variable}
%\end{table}
%%%%%%%%%%%%

%%%%%%%%%%%%
\subsection{Probability of Successful BCP Attack}%: $P_{B}$}
Recall that for successfully executing the BCP attack, $\mathcal{A}$ must have more than $c-\tau$ (Sybil) IDs in at least one shard. Hence, if $M \leq c - \tau$, $\mathcal{A}$ cannot launch BCP attacks, i.e., the probability of a successful BCP attack would be zero. Thus, to calculate the successful probability of a BCP attack, we first need to calculate the probability of having at least $c-\tau +1$ Sybil IDs in one shard, when $x=n$ Sybil IDs (generated by the adversary $\mathcal{A}$) have been chosen from the overall ID pool (at the end of the initialization step). The following lemma captures this probability.
\begin{lemma} 
In a shard-based blockchain protocol, when $n$ Sybil IDs (generated by the adversary $\mathcal{A}$) have been chosen from the ID pool after the initialization step, the probability of having at least $c-\tau +1$ Sybil IDs in one shard is: 
\begin{equation}\nonumber
P_{c-\tau +1}=\frac{2^{s}\sum\limits_{m=c-\tau +1}^{n}{n \choose m}{{N^*}-n \choose c-m}}{\binom {N^*}c}.
\end{equation}
where $c-\tau +1 \leq n$. 
\label{P(c-tau+1)}
\end{lemma}
\begin{proof}
Given the number of shards (i.e., $2^s$), the capacity of each shard (i.e., $c$), and the total number of selected IDs (i.e., $N^*$) the sample space and  the space of the desirable event (i.e., having at least $c-\tau +1$ Sybil IDs in one shard) will be:% calculated by:

\begin{equation}\label{total}
\fontsize{8}{5}
n(S)={{N^*} \choose c}{{N^*}-c \choose c}...{c \choose c}=\frac{{N^*}!}{c! c!...c!}=\frac{{N^*}!}{{c!}^{2^{s}}}.
\end{equation}
%Similarly, by considering all possible combinations of Sybil ID distributions in shards, we can compute the space of the desirable event (i.e., having at least $c-\tau +1$ Sybil IDs in one shard) by: 
%\begin{footnotesize}

\begin{equation}
\fontsize{8}{5}
\begin{aligned}
n(E)=& 2^{s}{n \choose c-\tau +1}{{N^*}-n \choose \tau -1}{{N^*}-c \choose c}..{c \choose c}\\&+2^{s}{n \choose c-\tau +2}{{N^*}-n \choose \tau -2}{T-c \choose c}..{c \choose c}\\&+...+2^{s}{n \choose n}{{N^*}-n \choose c-n}{{N^*}-c \choose c}..{c \choose c}.
\end{aligned} 
\label{n(E1)}
\end{equation}
%\end{footnotesize}
The probability $P_{c-\tau +1}$ can then be computed by:% taking the ratio of the event space to the sample space as:
\begin{equation}\nonumber
\fontsize{8}{5}
\begin{aligned}
P_{c-\tau +1}=\frac{n(E)}{n(S)}=\frac{2^{s}\sum\limits_{m=c-\tau +1}^{n}{n \choose m}{{N^*}-n \choose c-m}}{{{N^*}\choose c}}.
\end{aligned}
\end{equation}
\end{proof} 
% \remark{We need to say a few words what we can do for the case where n bigger than c. We cannot find close solution for this.}
\vspace{-10pt}
Please note that the above closed-form solution is correct if $n \leq c$. But for the values of $n$ bigger than $c$ we cannot employ the same calculation and deriving a closed-form solution was not possible. In the following analysis we employ numerical analysis to obtain these values, and we leave this derivation for our future work. The following theorem calculates the successful probability of BCP attack, when the number of Sybil IDs is smaller than or equal to $c$.
\begin{theorem}\label{thm:BCP1}
In a shard-based blockchain protocol, when the number of Sybil IDs (generated by the adversary $\mathcal{A}$) is smaller than or equal to $c$ and greater than $c-\tau$, the probability of a successful BCP attack in at least one shard is: 
\begin{equation}
P_{B}=\frac{2^{s}\sum\limits_{n=c- \tau +1}^{M}\sum\limits_{m=c- \tau +1}^{n}{M \choose n}{N-1 \choose {N^*}-n}{n \choose m}{{N^*}-n \choose c-m}}{{M+N-1 \choose {N^*}}{{N^*} \choose c}}.
\end{equation}
\end{theorem}
\begin{proof}
%After selecting ${N^*}$ IDs to distribute among shards, we can calculate the probability of having at least $c-\tau +1$ IDs of the attacker (i.e. $n$) in at least one shard where $(c-\tau) < n \leq c$.
The probability of a successful BCP attack in at least one shard can be computed as:

\begin{equation}
\fontsize{8}{5}
\begin{aligned}
P_{B}&=P\{x=c-\tau +1\}P_{c-\tau +1} + P\{x=c-\tau +2\}P_{c-\tau +1}\\&+...+P\{x=M\}P_{c-\tau +1} =(\sum\limits_{n=c-\tau +1}^{M}P\{x=n\})P_{c-\tau +1}.
\end{aligned}
\label{thm:BCP1-1}
\end{equation} 
%By elaborating $Prob\{.\}$ and $P_{c-\tau +1}$ 
Given Lemma \ref{pr} and Lemma \ref{P(c-tau+1)}, we can rewrite $P_B$ by:
\begin{equation}\nonumber
\fontsize{8}{5}
\begin{aligned}
&P_{B}=(\frac{\sum\limits_{n=c-\tau +1}^{M}{M \choose n}{N-1 \choose {N^*}- n}}{{M+N-1 \choose {N^*}}})(\dfrac{2^{s}\sum\limits_{m=c-\tau +1}^{n}{n \choose m}{{N^*}-n \choose c-m}}{{{N^*}\choose c}})\\ &=\frac{2^{s}\sum\limits_{n=c-\tau +1}^{M}\sum\limits_{m=c-\tau +1}^{n}{M \choose n}{N-1 \choose {N^*}- n}{n \choose m}{{N^*}-n \choose c-m}}{{M+N-1 \choose {N^*}}{{N^*}\choose c}}.
\end{aligned}
\end{equation}
\end{proof}
Having computed this probability, we now attempt to determine the impact of different shard-based blockchain system design parameters on the robustness against such Sybil attacks. In the following two theorems, we present similar success probability estimations for higher values of $M$.
\begin{theorem}\label{thm:BCP2}
In a shard-based blockchain protocol, when the number of Sybil IDs (generated by the adversary $\mathcal{A}$) is smaller than or equal to $2^s(c-\tau)$ and greater than $c$, the probability of a successful BCP attack in at least one shard is: 
\begin{equation}
\fontsize{8}{5}
\begin{aligned}
P_{B} &=\frac{2^{s}\sum\limits_{n=c-\tau +1}^{c}\sum\limits_{m=c-\tau +1}^{n}{M \choose n}{N-1 \choose {N^*}- n}{n \choose m}{{N^*}-n \choose c-m}}{{M+N-1 \choose {N^*}}{{N^*}\choose c}} \\
&+\frac{\sum\limits_{n=c+1}^{M}{M\choose n}{N-1\choose {N^*}-n}}{{M+N-1\choose {N^*}}}P'_{c-\tau +1}.
\end{aligned}
\end{equation}
\end{theorem}
\begin{proof}
Similar to the proof of Theorem~\ref{thm:BCP1}, this probability can be calculated by:
%\begin{equation}
%%\fontsize{8}{5}
%\nonumber
%\begin{aligned}
%P_{B}&=Prob\{x=c-\tau +1\}P_{c-\tau +1} \\&+ Prob\{x=c-\tau +2\}P_{c-\tau +1}\\&+...+Prob\{x=c\}P_{c-\tau +1}\\&+Prob\{x=c+1\}P'_{c-\tau +1}\\&+Prob\{x=c+2\}P'_{c-\tau +1}\\&+...+Prob\{x=M\}P'_{c-\tau +1}\\&=(\sum\limits_{n=c-\tau +1}^{c}Prob\{x=n\})P_{c-\tau +1}\\&+(\sum\limits_{n=c+1}^{M}Prob\{x=n\})P'_{c-\tau +1}.
%\end{aligned}
%\end{equation} 
\begin{equation}
\fontsize{8}{5}
\nonumber
\begin{aligned}
P_{B}&=(\sum\limits_{n=c-\tau +1}^{c}P\{x=n\})P_{c-\tau +1}+(\sum\limits_{n=c+1}^{M}P\{x=n\})P'_{c-\tau +1}.
\end{aligned}
\end{equation}

And we can replace $P_{c-\tau +1}$ to find the probability of successful attack by:% \remark{Again we should mention that having a closed form solution for $P'_{c-\tau +1}$ is quite hard.}
\begin{equation}
\fontsize{8}{5}
\nonumber
\begin{aligned}
P_{B}=&\frac{2^{s}\sum\limits_{n=c-\tau +1}^{c}\sum\limits_{m=c-\tau +1}^{n}{M \choose n}{N-1 \choose {N^*}- n}{n \choose m}{{N^*}-n \choose c-m}}{{M+N-1 \choose {N^*}}{{N^*}\choose c}}\\&+\frac{\sum\limits_{n=c+1}^{M}{M\choose n}{N-1\choose {N^*}-n}}{{M+N-1\choose {N^*}}}P'_{c-\tau +1}.
\end{aligned}
\end{equation} 
%
%If the $P'_{c-\tau +1}$ is computed, the formula of $P_{B}$ will be completed.
\end{proof}
\begin{theorem}\label{thm:BCP3}
In a shard-based blockchain protocol, when the number of Sybil IDs (generated by the adversary $\mathcal{A}$) is greater than $2^s(c-\tau)$, the probability of a successful BCP attack in at least one shard is: 
\begin{equation}
\begin{aligned}
&P_{B} =\frac{2^{s}\sum\limits_{n=c-\tau +1}^{c}\sum\limits_{m=c-\tau +1}^{n}{M \choose n}{N-1 \choose {N^*}- n}{n \choose m}{{N^*}-n \choose c-m}}{{M+N-1 \choose {N^*}}{{N^*}\choose c}}\\
&+\frac{\sum\limits_{n=c+1}^{2^{s}(c-\tau)}{M\choose n}{N-1\choose {N^*}-n}}{{M+N-1\choose {N^*}}}P'_{c-\tau +1}+\frac{\sum\limits_{n=2^{s}(c-\tau)+1}^{min(M,{N^*})}{M\choose n}{N-1\choose {N^*}-n}}{{M+N-1\choose {N^*}}}.
\end{aligned}
\end{equation}
\end{theorem}

\begin{proof}
Similar to earlier proofs, we can compute this probability by:
\begin{equation}
\fontsize{8}{5}
\begin{aligned}
%P_{B}&=Prob\{x=c-\tau +1\}P_{c-\tau +1} \\&+ Prob\{x=c-\tau +2\}P_{c-\tau +1}\\&+...+Prob\{x=c\}P_{c-\tau +1}\\&+Prob\{x=c+1\}P'_{c-\tau +1}\\&+Prob\{x=c+2\}P'_{c-\tau +1}\\&+...+Prob\{x=2^{s}(c-\tau)\}P'_{c-\tau +1}\\&+Prob\{x=2^{s}(c-\tau)+1\}p''_{c-\tau +1}\\&+Prob\{x=2^{s}(c-\tau)+2\}P''_{c-\tau +1}\\&+...+Prob\{x=min(M,{N^*})\}P''_{c-\tau +1}\\
P_B &=(\sum\limits_{n=c-\tau +1}^{c}P\{x=n\})P_{c-\tau +1}+(\sum\limits_{n=c+1}^{2^{s}(c-\tau)}P\{x=n\})P'_{c-\tau +1}\\&+(\sum\limits_{n=2^{S}(c-\tau)+1}^{min(M,{N^*})}P\{x=n\})P''_{c-\tau +1}.
\end{aligned}
\label{thm:BCP3-1}
\end{equation}
Since $n > 2^{s}(c-\tau )$, $c-\tau$ Sybil IDs (generated by the adversary) will definitely reside in each shard. In other words, given the total number of Sybil IDs generated by the adversary, at least $c-\tau +1$ of these IDs will be placed in at least one shard. When this happens, $\mathcal{A}$ can compromise the consensus protocol of this shard with probability 1. So $P''_{c-\tau +1}=1$. % and we can rewrite Equation~\eqref{thm:BCP3-1} as follows:
%
%\begin{equation}
%%\fontsize{8}{5}
%\begin{aligned}
%P_{B}&=(\sum\limits_{n=c-\tau +1}^{c}Prob\{x=n\})P_{c-\tau +1}\\&+(\sum\limits_{n=c+1}^{2^{s}(c-\tau)}Prob\{x=n\})P'_{c-\tau +1}\\&+\sum\limits_{n=2^{S}(c-\tau)+1}^{min(M,{N^*})}Prob\{x=n\}.
%\end{aligned}
%\label{thm:BCP3-2}
%\end{equation}
Moreover based on Lemma \ref{pr} and Lemma \ref{P(c-tau+1)}, Equation \eqref{thm:BCP3-1} can be rewritten as:
%Similarly, by elaborating $Prob\{x=n\}$ and $P_{c-\tau +1}$ in  $P_{B}$ will be equal to:
\begin{equation}\nonumber
\fontsize{8}{5}
\begin{aligned}
&P_{B}=\frac{2^{s}\sum\limits_{n=c-\tau +1}^{c}\sum\limits_{m=c-\tau +1}^{n}{M \choose n}{N-1 \choose {N^*}- n}{n \choose m}{{N^*}-n \choose c-m}}{{M+N-1 \choose {N^*}}{{N^*}\choose c}}\\&+\frac{\sum\limits_{n=c+1}^{2^{s}(c-\tau)}{M\choose n}{N-1\choose {N^*}-n}}{{M+N-1\choose {N^*}}}P'_{c-\tau +1}+\frac{\sum\limits_{n=2^{s}(c-\tau)+1}^{min(M,{N^*})}{M\choose n}{N-1\choose {N^*}-n}}{{M+N-1\choose {N^*}}}.
\end{aligned}
\end{equation}
\end{proof}
%
%Our analytical results show that we can compute such Sybil attacks to consensus protocol in shard-based blockchain protocols. 
%MJ-The above sentence is not clear. What do you mean by "we can compute such Sybil attacks"
\vspace{-12pt}
In the following, we will similarly analyze the success probability of an adversary in executing the GFT attack. 
%
%\vspace{-10pt}
\subsection{Probability of a Successful GFT Attack}
In GFT attack, $\mathcal{A}$ would like to change the output of at least one shard, and insert his favorite transactions in the blockchain. The calculation of successful probability of this attack  is similar to the previous section (i.e., BCP attack), but the number of required Sybil IDs is different. In this attack, $\mathcal{A}$ must have more than $\tau$ Sybil IDs in at least one shard. Hence, if $M < \tau$, the adversary cannot successfully execute the GFT attacks (i.e., probability of success is zero). In order to estimate the success probability of an adversary in carrying out GFT attacks, we first need to calculate the probability of having at least $\tau$ Sybil IDs in one shard, when a total of $x=n$ Sybil IDs (generated by the adversary) have been chosen from the pool of all generated IDs (during the initialization step). The following lemma calculates this probability.
\begin{lemma} 
In a shard-based blockchain protocol, when $n$ Sybil IDs (generated by the adversary $\mathcal{A}$) have been chosen or selected from the ID pool after the initialization step, the probability of having at least $\tau $ Sybil IDs in one shard is:
\begin{equation}\nonumber
\fontsize{8}{6}
P_\tau =\frac{2^{s}\sum\limits_{m=\tau}^{n}{n \choose m}{{N^*}-n \choose c-m}}{{{N^*}\choose c}}.
\end{equation}
where $\tau \leq n$. 
\label{P tau}
\end{lemma}
%\vspace{-0.25cm}
\begin{proof}
According to the Lemma \ref{P(c-tau+1)}, the sample space is equal to Equation~\eqref{total}.
Similarly, by considering all possible combinations of Sybil ID distributions within the shards, we can compute the space of the desirable event (i.e., having at least $\tau $ Sybil IDs in one shard) by:
%\vspace{-0.25cm}
\begin{equation}\nonumber
\fontsize{8}{6}
\begin{aligned}
n(E)=& 2^{s}{n \choose \tau}{{N^*}-n \choose c-\tau}{{N^*}-c \choose c}..{c \choose c}\\&+2^{s}{n \choose \tau +1}{{N^*}-n \choose c-\tau -1}{{N^*}-c \choose c}..{c \choose c}\\&+...+2^{s}{n \choose n}{{N^*}-n \choose c-n}{{N^*}-c \choose c}..{c \choose c}
\end{aligned} 
\label{n(E2)}
\end{equation} 
The probability $P_\tau $ can then be computed by taking the ratio of the event space to the sample space as:
\begin{equation}
P_\tau =\dfrac{n(E)}{n(S)}=\dfrac{2^{s}\sum\limits_{m=\tau}^{n}{n \choose m}{{N^*}-n \choose c-m}}{{{N^*}\choose c}}.
\end{equation}
\end{proof}
\vspace{-6pt}
Similar to BCP attack, the above closed-form solution is correct if $n \leq c$. Finding a closed-form solution for the case where $n > c$ is a part of our future challenges. 
% But for the values of $n$ bigger than $c$ we cannot employ the same calculation. \remark{Same Problem as Previous Lemma} 
%
The following theorem calculates the successful probability of GFT attack, when the number of Sybil IDs is smaller than or equal to $c$.

\begin{theorem}\label{thm:GFT1}
In a shard-based blockchain protocol, when the number of Sybil IDs (generated by the adversary $\mathcal{A}$) is smaller than or equal to $c$ and greater than $\tau -1$, the probability of a successful GFT attack in at least one shard is: 
\begin{equation}
P_{G}=\frac{2^{s}\sum\limits_{n=\tau}^{M}\sum\limits_{m=\tau }^{n}{M \choose n}{N-1 \choose {N^*}-n}{n \choose m}{{N^*}-n \choose c-m}}{{M+N-1 \choose {N^*}}{{N^*} \choose c}}.
\end{equation}
\end{theorem}
\begin{proof}
This probability can be computed similar to  the proof of Theorem~\ref{thm:BCP1} as:
\begin{equation}
\fontsize{8}{5}
\begin{aligned}
%P_{G}&=P\{x=\tau\}P_\tau + P\{x=\tau +1\}P_\tau \\&+...+P\{x=M\}P_\tau =(\sum\limits_{n=\tau}^{M}P\{x=n\})P_\tau.
P_{G}&=(\sum\limits_{n=\tau}^{M}P\{x=n\})P_\tau.
\end{aligned}
\label{thm:GFT1-1}
\end{equation} 
Based on Lemma \ref{pr} and Lemma \ref{P tau}, Equation~\eqref{thm:GFT1-1} can be rewritten as:
%By elaborating $Prob\{x=n\}$ and $P_\tau$ in Equation~\eqref{thm:GFT1-1}, that was calculated  in Lemma \ref{pr} and Lemma \ref{P tau} respectively, we obtain:
\begin{equation}\nonumber
\fontsize{8}{5}
\begin{aligned}
P_{G}&=(\dfrac{\sum\limits_{n=\tau}^{M}{M \choose n}{N-1 \choose {N^*}- n}}{{M+N-1 \choose {N^*}}})(\dfrac{2^{s}\sum\limits_{m=\tau}^{n}{n \choose m}{{N^*}-n \choose c-m}}{{{N^*}\choose c}}) \\& =\dfrac{2^{s}\sum\limits_{n=\tau}^{M}\sum\limits_{m=\tau}^{n}{M \choose n}{N-1 \choose {N^*}- n}{n \choose m}{{N^*}-n \choose c-m}}{{M+N-1 \choose {N^*}}{{N^*}\choose c}}.
\end{aligned}
\end{equation}
\end{proof}
Having computed this probability, we now attempt to determine the impact of different shard-based blockchain system design parameters on the robustness against such Sybil-based GFT attacks. In the following two theorems, we present similar success probability estimations for higher values of $M$.
\begin{theorem}\label{thm:GFT2}
In a shard-based blockchain protocol, when the number of Sybil IDs (generated by the adversary $\mathcal{A}$) is smaller than or equal to $2^s(\tau -1)$ and greater than $c$, the probability of a successful GFT attack in at least one shard is: 
\begin{equation}\label{Eqn_PG}
\fontsize{8}{5}
\begin{aligned}
P_{G} &=\frac{2^{s}\sum\limits_{n=\tau}^{c}\sum\limits_{m=\tau}^{n}{M \choose n}{N-1 \choose {N^*}- n}{n \choose m}{{N^*}-n \choose c-m}}{{M+N-1 \choose {N^*}}{{N^*}\choose c}}+\frac{\sum\limits_{n=c+1}^{M}{M\choose n}{N-1\choose {N^*}-n}}{{M+N-1\choose {N^*}}}P'_\tau .
\end{aligned}
\end{equation}
\end{theorem}
\begin{proof}
Similar to the proof of Theorem~\ref{thm:GFT1}, this probability can be calculated by:
\begin{equation}\nonumber
\fontsize{8}{5}
\begin{aligned}
P_{G}&= (\sum\limits_{n=\tau}^{c}P\{x=n\})P_\tau+(\sum\limits_{n=c+1}^{M}P\{x=n\})P'_\tau .
\end{aligned}
\label{thm:GFT2-1}
\end{equation} 
It is then easy to show that $P_G$ is derived by Equation~(\ref{Eqn_PG}).
% And we can replace $P_{\tau}$ from Lemma~\ref{P tau} to find the probability of successful attack by% \remark{Again we should mention that having a closed form solution for $P'_{\tau}$ is quite hard.}
%\begin{equation}\nonumber
%\fontsize{8}{5}
%\begin{aligned}
%P_{G}&=\dfrac{2^{s}\sum\limits_{n=\tau}^{c}\sum\limits_{m=\tau}^{n}{M \choose n}{N-1 \choose T- n}{n \choose m}{{N^*}-n \choose c-m}}{{M+N-1 \choose {N^*}}{{N^*}\choose c}}+\dfrac{\sum\limits_{n=c+1}^{M}{M\choose n}{N-1\choose{N^*} -n}}{{M+N-1\choose {N^*}}}P'_\tau .
%\end{aligned}
%\end{equation} 
\end{proof}
\begin{theorem}\label{thm:GFT3}
In a shard-based blockchain protocol, when the number of Sybil IDs (generated by the adversary $\mathcal{A}$) is greater than $2^s(\tau -1)$, the probability of a successful GFT attack in at least one shard is: 
\begin{equation}\label{eqn_GFT3}
\fontsize{8}{5}
\begin{aligned}
P_{G} &=\frac{2^{s}\sum\limits_{n=\tau}^{c}\sum\limits_{m=\tau}^{n}{M \choose n}{N-1 \choose {N^*}- n}{n \choose m}{{N^*}-n \choose c-m}}{{M+N-1 \choose {N^*}}{{N^*}\choose c}}+\frac{\sum\limits_{n=c+1}^{2^{s}(\tau -1)}{M\choose n}{N-1\choose {N^*}-n}}{{M+N-1\choose {N^*}}}P'_\tau \\&+\frac{\sum\limits_{n=2^{s}(\tau -1)+1}^{min(M,{N^*})}{M\choose n}{N-1\choose {N^*}-n}}{{M+N-1\choose {N^*}}}.
\end{aligned}
\end{equation}
\end{theorem}

\begin{proof}
Similarly, the probability of a successful GFT attack can be computed as:
\begin{equation}
\fontsize{8}{5}
\begin{aligned}
%&P_{G}=Prob\{x=\tau\}P_\tau + Prob\{x=\tau +1\}P_\tau \\&+...+Prob\{x=c\}P_\tau \\&+Prob\{x=c+1\}P'_\tau +Prob\{x=c+2\}P'_\tau \\&+...+Prob\{x=2^{s}(\tau -1)\}P'_\tau \\&+Prob\{x=2^{s}(\tau-1)+1\}P''_\tau \\&+Prob\{x=2^{s}(\tau -1)+2\}P''_\tau \\&+...+Prob\{x=min(M,{N^*})\}P''_\tau .
%\end{aligned}
%\end{equation}
%\begin{equation}
%\fontsize{8.6}{6}
%\begin{aligned}
P_{G}&=(\sum\limits_{n=\tau}^{c}P\{x=n\})P_\tau+(\sum\limits_{n=c+1}^{2^{s}(\tau -1)}P\{x=n\})P'_\tau\\&+(\sum\limits_{n=2^{s}(\tau-1)+1}^{min(M,{N^*})}P\{x=n\})P''_\tau .
\end{aligned}
\label{thm:GFT3-1}
\end{equation}
Since $n > 2^{s}(\tau -1 )$, $\tau -1$ Sybil IDs (generated by the adversary) will definitely reside in each shard. In other words, given the total number of Sybil IDs generated by the adversary, at least $\tau$ of these IDs will be placed in at least one shard. When this happens, $\mathcal{A}$ can compromise the PBFT consensus protocol to change the output produced by this shard with probability 1. So $P''_{\tau}=1$ and we can rewrite Equation~\eqref{thm:GFT3-1} as shown by Equation~(\ref{eqn_GFT3}).
%\begin{equation}\nonumber
%\fontsize{8}{5}
%\begin{aligned}
%&P_{G}=\dfrac{2^{s}\sum\limits_{n=\tau}^{c}\sum\limits_{m=\tau}^{n}{M \choose n}{N-1 \choose {N^*}- n}{n \choose m}{{N^*}-n \choose c-m}}{{M+N-1 \choose {N^*}}{{N^*}\choose c}}\\&+(\dfrac{\sum\limits_{n=c+1}^{2^{s}(\tau -1)+1}{M\choose n}{N-1\choose {N^*}-n}}{{M+N-1\choose {N^*}}})P'_\tau +\dfrac{\sum\limits_{n=2^{s}(\tau -1)+1}^{min(M,{N^*})}{M\choose n}{N-1\choose {N^*}-n}}{{M+N-1\choose {N^*}}}.
%\end{aligned}
%\end{equation}
\end{proof}
In the following section, we will verify our closed-form solutions with  the help of numerical simulations. %and then by some numerical analysis of Sybil attacks in shard-based blockchains.
%%%%%%%%%%%%%%%
\begin{figure*}[t]
\centering
\begin{subfigure}[t]{0.3\textwidth}
\includegraphics[width=\linewidth]{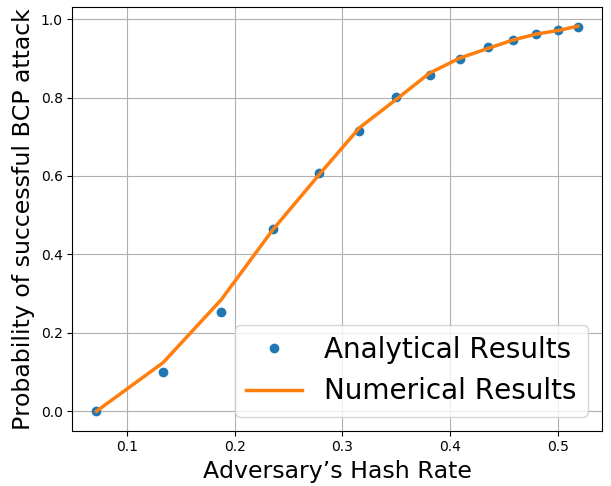}
\caption{}
\label{BCP verification 14 nodes}
\end{subfigure}
\hspace{12pt}
\begin{subfigure}[t]{0.3\textwidth}
\includegraphics[width=\linewidth]{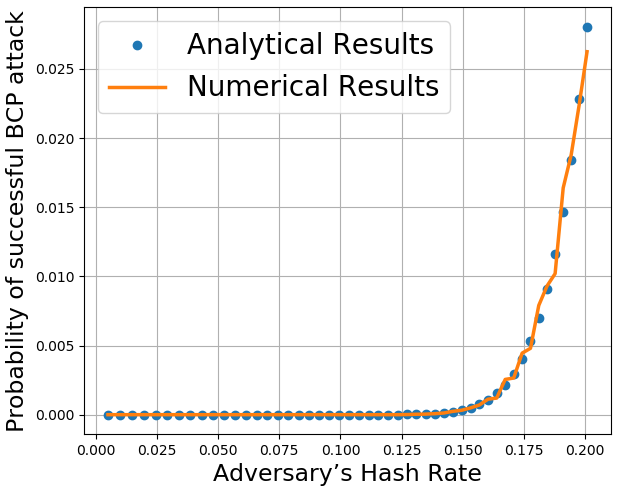}
\caption{}
\label{BCP verification 200 nodes}
\end{subfigure}
\hspace{12pt}
\begin{subfigure}[t]{0.3\textwidth}
\includegraphics[width=\linewidth]{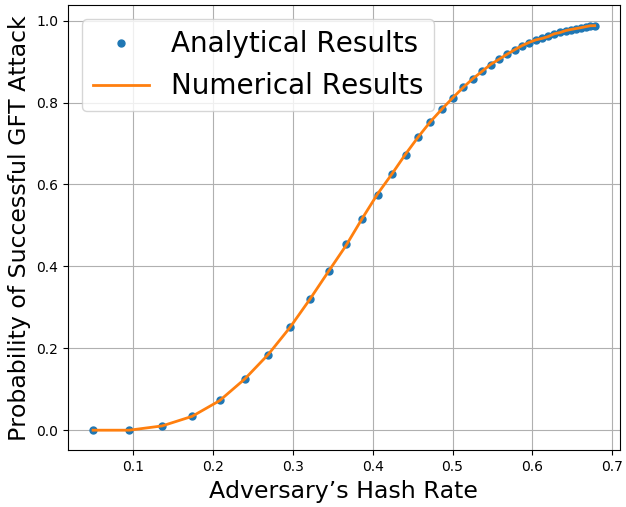}
\caption{}
\label{GFT verification 20 nodes}
\end{subfigure}
\caption{Our model verification. (a) BCP attack simulations  for $N=14$ with 4 shards. (b) BCP attack simulation with $N=200$ and 4 shards. (c) GFT attack simulation with $N=20$ and 4 shards.}
\label{verification}
\end{figure*}
%%%%%%%%%%%%%%
%%%%%%%%%%%%%%%%%%%%%%%%%%%%%%%%%%%%
\section{Numerical Results and Model Verifications}
\label{sec:numerical}

%In order to verify the relevance of our proposed models to compute the successful probability of BCP and GFT attacks, we compare 
We validate the correctness of our analytical results discussed earlier by conducting extensive numerical simulations as outlined next. We implement a Python-based simulator for Elastico~\cite{luu2016secure} and simulate the BCP and GFT Sybil attacks by considering an adversary with different hash-powers. Success probabilities of these BCP and GFT attacks during simulations is computed and compared with our previous determined analytical results. We investigate the effect of various system parameters on the success probability of these attacks, including, number of shards (i.e., $2^s$), capacity of each shard (i.e., $c$), total number of participating nodes in the network (i.e., $N$) and the threshold for consensus (i.e., $\tau$). 
%We simulate a network of 1000 nodes, where 999 nodes are honest (i.e., generate just one ID) and the one adversary ($\mathcal{A}$) who has the capability of generating multiple sybil IDs.
%MJ-this last line is not clear. In the previous sentence we said N is a parameter we vary, but then in the next line we say simulation consists of N=1000

%%%%%%%%%%%%%%%
\begin{figure*}[t!]
\centering
\begin{subfigure}[t]{0.24\textwidth}
\includegraphics[width=\linewidth]{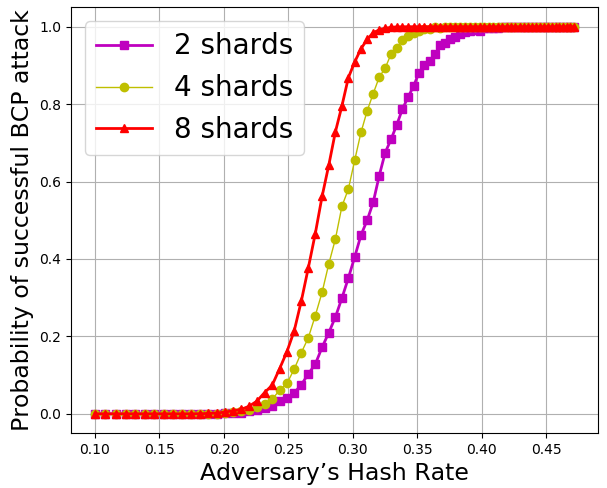}
\caption{}%The number of shards, for 100 member shards and $N=1000$.}
\label{fig_sim_num_shard_1}
\end{subfigure}
\begin{subfigure}[t]{0.24\textwidth}
\includegraphics[width=\linewidth]{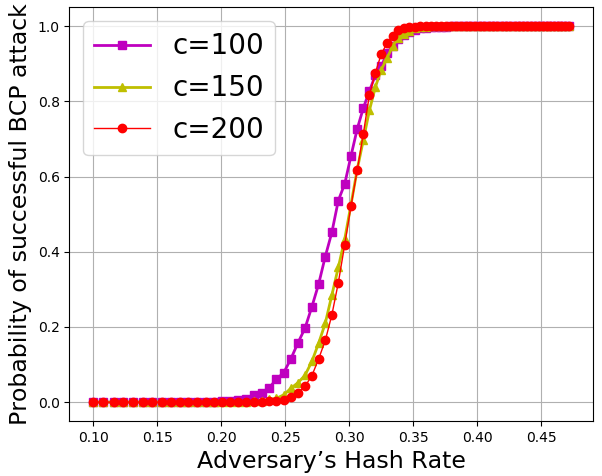}
\caption{}%The capacity of every shard, for 4 shards and $N=1000$.}
\label{fig_sim_cap_shard_1}
\end{subfigure}
\begin{subfigure}[t]{0.24\textwidth}
\includegraphics[width=\linewidth]{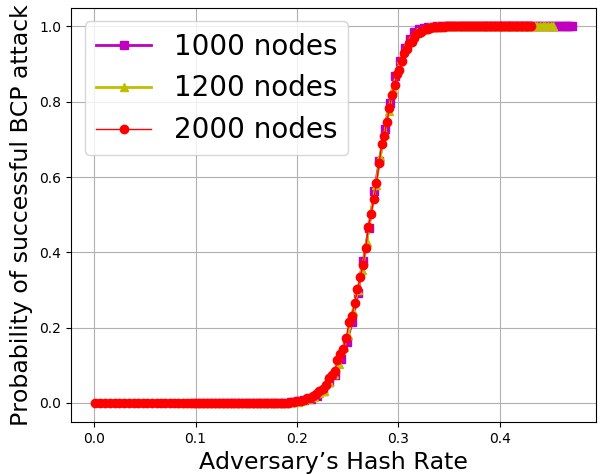}
\caption{}%The number of total nodes, for 8 100 member shards.}
\label{fig_sim_num_node_1}
\end{subfigure}
\begin{subfigure}[t]{0.24\textwidth}
\includegraphics[width=\linewidth]{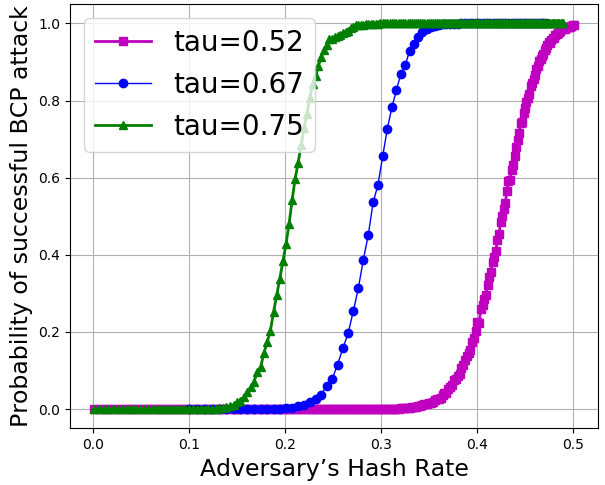}
\caption{}%The $\tau$ value, for 4 100 member shards and $N=1000$.}
\label{fig_sim_tau_1}
\end{subfigure}
\caption{The effect of (a) number of shards, (b) capacity of shard, (c) number of nodes, and (d)  $\tau$  on successful BCP attack.}
\label{BCP Attack}
\end{figure*}
%%%%%%%%%%%%%%%

%%%%%%%%%%%%%%
\subsection{Validation of Analytical Results}

%\begin{figure}[t]
%\centering
%\includegraphics[width=.5\linewidth]{equation_VS_simulation_GFT_20node.png}
%%\caption{Model verification for 4 4 member shard and $N=20$}
%\caption{GFT model verification.}
%\label{GFT verification}
%\end{figure}
%MHMWe first validate our analytical results from Section \ref{sec:analytical} by means of simulations.
%verified our presented results in Theorems~\ref{thm:BCP1}, \ref{thm:BCP2} and \ref{thm:BCP3} by running some numerical simulations.  
We start by considering a small network of 14 participating nodes (i.e., $N=14$) and 4 shards (i.e., $s=2$), each with a capacity of 3 nodes (i.e., $c=3$) undergoing a BCP attack by assuming different adversarial hash-rates. 
%The probability of successful BCP attack is calculated by Theorems~\ref{thm:BCP1},\ref{thm:BCP2}, and \ref{thm:BCP3}.
Figure~\ref{BCP verification 14 nodes}, which represents the probability of a successful BCP attack under this simulated scenario, shows that results from our simulations align very closely to those obtained from our analytical results (i.e., Theorems~\ref{thm:BCP1}, \ref{thm:BCP2} and \ref{thm:BCP3}). Here, the hash-rate (shown on the x-axis) is computed as the ratio of the adversary's hash-power ($\mathcal{A}$) to the average hash power of the entire network. Figure~\ref{BCP verification 200 nodes} shows similar results for a larger network with parameters $N=2
00$, $c=50$, and $s=2$. Even in this case, it can be observed that our analytical results are in line with the simulated behavior of the Elastico network.
We similarly verify the validity of our analytical results (Theorems~\ref{thm:GFT1}, \ref{thm:GFT2} and \ref{thm:GFT3}) for the GFT attacks. Figure \ref{GFT verification 20 nodes} shows the success probability of the GFT attacks in an Elastico blockchain network with parameters $N=20$, $c=4$, and $s=2$, and these results confirm that our analytical computation of these probabilities was correct.
%where our model again follows the real scenarios.
%We simulated a blockchain network again and set the parameters equal to .   
Next, for each attack scenario, we will employ simulations to demonstrate the impact of different system parameters on the success probability of the attacks.

\subsection{BCP Numerical Analysis}

Figure~\ref{BCP Attack} shows the success probability of BCP attacks for different hash-powers of the adversary. Figure~\ref{fig_sim_num_shard_1} shows that the BCP attack probability increases when we increase the number of shards. In this experiment, the capacity ($c$) of each shard is set to 100 (based on Elastico) and the value of $\tau $ is set to $\frac{2}{3}$. Our results show that an adversary who has 25\% of the hash-power of network can compromise (and manipulate) the consensus algorithm employed by the shard-based protocol (e.g., PBFT). Figure~\ref{fig_sim_cap_shard_1} shows that the probability of successful BCP attacks decreases when the capacity ($c$) of each shard increases. We also evaluated the effect of the number of active/participating nodes on the probability of successful BCP attacks. We execute the simulations by setting the shard capacity ($c$) to 100, number of shards ($s$) to 4 and the threshold value $\tau$ to $\frac{2}{3}$. Results from these simulations show that the adversary ($\mathcal{A}$) needs to have 33\% of the hash-power of whole network to lunch a successful BCP attack. But accumulating 33\% of the network's hash-power will become more difficult by increasing the number of total nodes. 
%MJ-For result we are referring to in the last line of the above paragraph, which figure are we referring to? Please mention. Also the statement is not clear. What do you mean by "will become more difficult by increasing the number of total nodes"

Finally, we investigate the effect of $\tau$ on the success probability of BCP attacks. The results are shown in Figure~\ref{fig_sim_tau_1}. We vary the value of $\tau$ from $0.52$ and $0.75$. Our results show that if the value of $\tau$ changes from $\tau _1$ to $\tau _2$, the adversary needs about $1-(\tau _2 - \tau _1)$ hash-power of the previous hash-power, in order to achieve the same attack success probability for $\tau_1$. For example, if the BCP attack success probability of an adversary with hash-power $HP$ was $P_B$ with $\tau =\frac{2}{3}$, then the adversary would need to accumulate a hash-power of $0.91 \times HP$ in order to achieve the same success probability ($P_B$) with $\tau = 0.75$. 

\begin{figure*}[!h]
\centering
\begin{subfigure}[t]{0.24\textwidth}
\includegraphics[width=\linewidth]{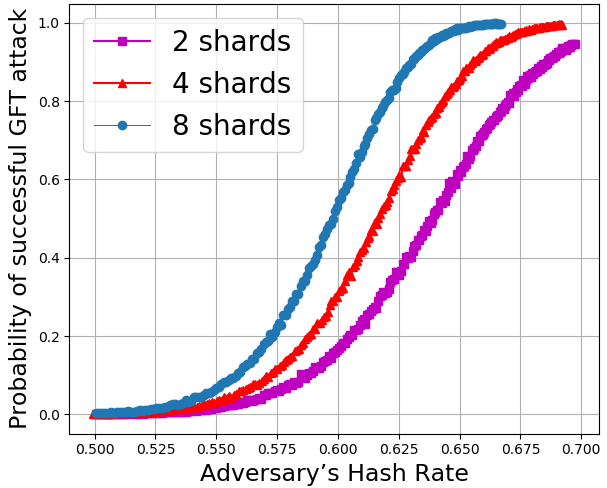}
\caption{}%The number of shards, for 100 member shards and $N=1000$.}
\label{fig_sim_num_shard_2}
\end{subfigure}
\begin{subfigure}[t]{0.24\textwidth}
\includegraphics[width=\linewidth]{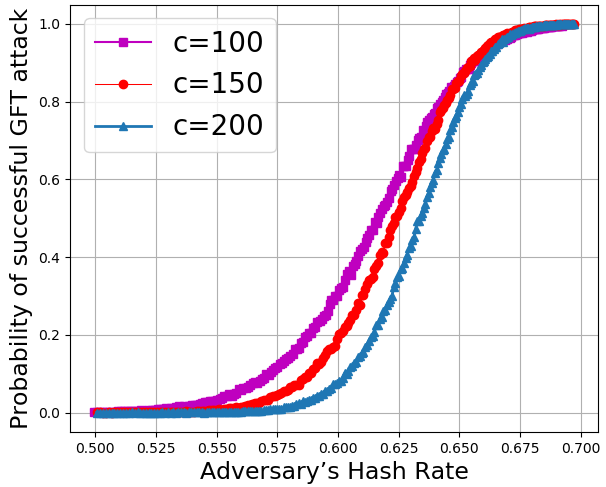}
\caption{}%The capacity of every shard, for 4 shards and $N=1000$.}
\label{fig_sim_cap_shard_2}
\end{subfigure}
\begin{subfigure}[t]{0.24\textwidth}
\includegraphics[width=\linewidth]{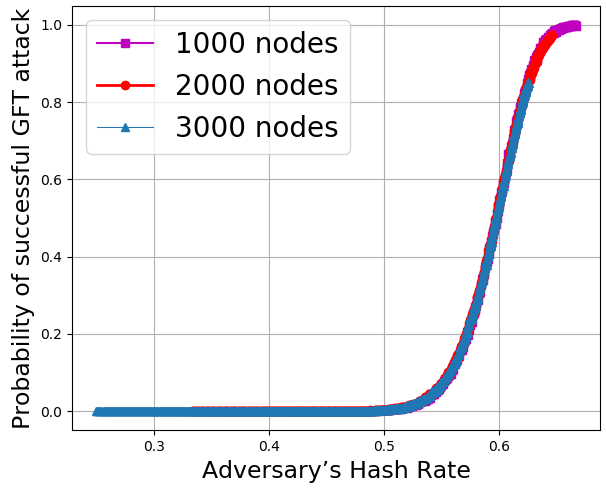}
\caption{}%The number of total nodes, for 8 100 member shards.}
\label{fig_sim_num_node_2}
\end{subfigure}
\begin{subfigure}[t]{0.24\textwidth}
\includegraphics[width=\linewidth]{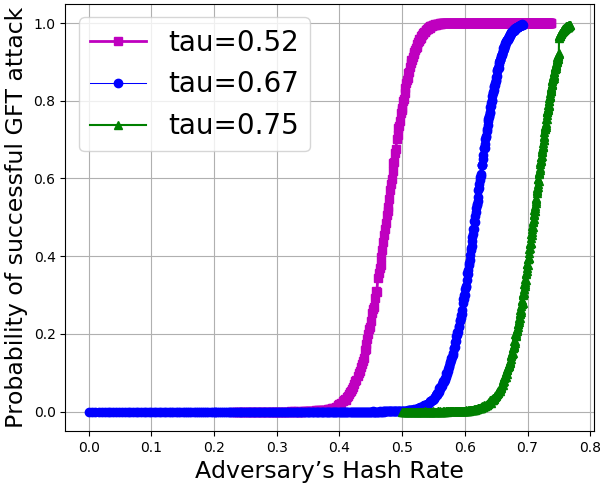}
\caption{}%The $\tau$ value, for 4 100 member shards and $N=1000$.}
\label{fig_sim_tau_2}
\end{subfigure}
\caption{The effect of (a) number of shards, (b) shard capacity, (c) number of nodes, and (d) $\tau$ on successful GFT attack.}%  some parameters on the probability of successful GFT attack. a) Number of shards, b) Capacity of every shard, c) Total nodes of the network and d) The $\tau$ value.}
\label{GFT Attack}
\end{figure*}
%%%%%%%%%%%%%

%%%%%%%%%%%%%
\subsection{GFT Numerical Analysis}

Figure~\ref{GFT Attack} shows the success probability of GFT attacks for different values of the adversary's hash-power. As shown in Figure~\ref{fig_sim_num_shard_2}, if the adversary's hash-power is less than 50\% of the network hash-power, it cannot execute a successful GFT attack. We also observe that the GFT attack probability increases when the number of shards increases. Moreover, Figure~\ref{fig_sim_cap_shard_2} shows that shard capacity ($c$) increases, the GFT attack success probability decreases. It should be noted that this trend is observable only when the adversary's hash-power is less than 65\% of the total network hash-power. 
 If the adversary has more than 65\% of the total network hash-power, increasing the capacity ($c$) has no significant impact on $P_G$.%the GFT attack probability.

Figure~\ref{fig_sim_num_node_2} shows the effect of the total number of participating nodes in the network on the success of the GFT attack. Our simulation results show that if the adversary holds 65\% of the network's hash-power, it can successfully perform the GFT attack when the total number of participating nodes $2000$ or less. However, as the total number of participating nodes increases (e.g., $3000$), the success of the GFT attack is not guaranteed.
%nodes the attacker cannot easily conduct a successful attack. 
Figure~\ref{fig_sim_tau_2} shows the impact of $\tau$ on the success of GFT attacks. Here we see a similar trend as in the case of BCP attacks, where with higher values of $\tau$ the adversary needs more hash-power to conduct a successful GFT attack.

Table~\ref{tab:Final} summarizes a series of simulation of Sybil attack with different parameters. The results show that we can avoid successful GFT and BCP attack with an optimal selection of number of shards and their capacity even if the adversary's hash-power is about 25\% of the average hash-power of the network. But if the adversary's hash-power is more than 33\%, it can successfully deploy BCP attack. However, with maximum of 16 shards with capacity of 600, we can decrease the probability of successful GFT to less than $0.001$.  We believe these results can help designers to avoid attacks.%, by carefully designing their network. 

%%%%%%%%%%%%%%
\begin{table}[t]
\caption{Numerical Evaluation of Sybil attacks.}\label{tab:Final}
\centering
\begin{tabular}{|c|c|c|c|c|c|}
\hline
\textbf{$h^{\mathcal{A}}$} & \textbf{$\tau$} & {\#shards}  & \textbf{$c$} & $P_{B}$ & $P_{G}$   \\ 
\hline
 $25\%$   &  $\frac{2}{3}$  & At most $16$   &  At least $600$  & $\leq 10^{-4}$  & 0 \\
\hline
 $[33\%, 53\%]$  &  $\frac{2}{3}$  & -   & -  & $\geq0.8$   & $\leq 0.005$ \\
\hline
  $56\%$  &   $\frac{2}{3}$  &  At most $16$  &  At least $600$  &  $1$  &  $\leq 0.001$\\
\hline
 Above $66\%$ &  $\frac{2}{3}$  & -   & -  & $1$  & $\geq 0.75$ \\
\hline
\end{tabular}
\end{table}
\section{Related Work}
\label{sec:ralated work}

To put our current research effort in perspective, we now briefly outline some other efforts in the literature towards improving the scalability and security of permissionless blockchains. Bitcoin-NG \cite{eyal2016bitcoin} was the first attempt to improve the transaction throughput of Bitcoin's consensus protocol \cite{nakamoto2008bitcoin} by employing the concept of \textit{microblocks}.
%MHM, where the selected leader node continues to append microblocks to the blockchain until a new leader is elected.
%. In Bitcoin-NG, similar to Bitcoin, a leader is selected randomly (by using PoW) at each epoch. 
Due to the significant drawbacks of leader based consensus protocols such as Bitcoin and Bitcoin-NG, the research community's focus shifted to employing committee-based consensus algorithms \cite{bracha1987log} for permissionless blockchain systems. For instance, Decker et al. \cite{decker2016bitcoin} proposed one of the first committee-based consensus protocols, named \emph{PeerCensus}, followed by several other proposals \cite{AbrahamMNRS16,pass2017hybrid,kogias2016enhancing,gilad2017algorand} in a similar direction. The poor scalability of single committee consensus solutions motivated the design of \emph{multiple committee} based blockchain consensus protocols, where the main idea is to split the pending transaction set among multiple shards, which then processes these shards in parallel. \emph{RSCoin} \cite{DanezisM:RSCoin} was proposed as a shard-based blockchain for centrally-banked cryptocurrencies, while \emph{Elastico} \cite{luu2016secure} was the first fully distributed shard-based consensus protocol for public blockchains. Recently proposed shard-based protocols, such as, \emph{Omniledger} \cite{kokoris2017omniledger} and \emph{Rapidchain} \cite{zamanirapidchain} attempt to improve upon the scalability and security guarantees of Elastico, while \emph{PolyShard} \cite{li2018polyshard} proposes to employ techniques from the \emph{coded computing} paradigm \cite{yu2018lagrange} to simultaneously improve transaction throughput, storage efficiency and security. Recently, several novel approaches for shard-based consensus protocols for blockchains have also been proposed \cite{dang2018towards,secure2018zilliqa,poon2017plasma}, and a good review of that various shard-based blockchain protocols can be found in \cite{wang2019sok}.
In the direction of security-related analysis and securing shard-based blockchain protocols, Jusik Yun et al.\cite{yun2019trust} observed that, as the number of validators per-shard in shard-based blockchain protocols is generally smaller than the number of validators in traditional single leader-based protocols, it makes shard-based protocols more vulnerable to 51\% attacks.
%than blockchains that do not use sharding. 
They then present a novel Trust-Based Shard Distribution scheme, which distributes the assignment of potential malicious nodes in the network to shards, in order to solve this problem.
%, which is a that assigns potential malicious nodes in the network to different shards, preventing malicious nodes from gaining a dominating influence on the consensus of a single shard. 
The vulnerability of blockchain consensus protocols, including shard-based protocols, to attacks that employ Sybil nodes \cite{douceur2002sybil} is clear and well-documented \cite{sankar2017survey,conti2018survey}. 
%So all of the blockchain schemes must be Sybil-resistante. In the traditional blockchain, the 
In this direction, TrustChain \cite{otte2017trustchain} proposed a novel Sybil-resistant algorithm, called NetFlow, which overcomes Sybil attacks in traditional single leader blockchain architectures by determining and employing node trustworthiness during the consensus process.
%of nodes. NetFlow ensures that nodes who take resources from the community also contribute back. They demonstrate that irrefutable historical transaction records offer security and seamless scalability, without requiring global consensus. Their experimentation shows that the transaction throughput of TrustChain surpasses that of traditional blockchain architectures like Bitcoin. 
%Researchers in \cite{sankar2017survey} believe that the blockchain is vulnerable to the Sybil-attack. Authors in \cite{conti2018survey} believe that the PoW protect of the blockchain against the Sybil attack. 
%%MHMIttay Eyal et al. in \cite{eyal2018majority} believed that the Sybil nodes can work like sensors. Also they would ignore some blocks and start propagating another block.
%MJ- The above sentence is not clear at all - please rewrite.
In the first mathematical efforts, researchers in \cite{hafid2019new} analyze the security of three shard-based blockchains (\cite{luu2016secure}, \cite{kokoris2018omniledger} and \cite{zamani2018rapidchain}). They computed the upper bound for the probability of increasing the number of malicious nodes for one committee and so for each epoch using tail inequalities for sums of bounded hypergeometric and binomial distributions. But they don't have any simulation to proof their probability results. 
%MJ-The above sentence needs to be re-written. Also, you need to explain clearly how this effort (\cite{hafid2019new}) is different from ours.

\section{Conclusion}
\label{sec:conclusion}

In this paper, we presented an analytical model to calculate the probability of successful Sybil attack to shard-based permissionless blockchains. As an example, we modeled Elastico and defined two types of Sybil attacked, named BCP (Break Consensus Protocol) and GFT (Generate Fake Transaction). We showed that we can calculate the probability of successful attacks given different system parameters, such as number of shards, the capacity of shards, and the number of nodes in the blockchain network. The results have been verified with numerical simulations. We showed that by carefully design our blockchain network we can avoid Sybil attack. Our results showed that Elastico is not robust against a Sybil attack, with an adversary who has 25\% hash-power of the network. In this case the adversary can break the consensus protocol in at least on shard with probability equal to 0.2.

\balance

% references section
\bibliographystyle{IEEEtran}
\bibliography{ref1}

% that's all folks
\end{document}